\documentclass[onecolumn,12pt]{IEEEtran}

\usepackage[utf8]{inputenc} 
\usepackage[T1]{fontenc}
\usepackage{url}              
\usepackage{cite}             

\usepackage[cmex10]{amsmath}  
\usepackage{mleftright}       
\mleftright                   

\usepackage{graphicx}         
\usepackage{color}                              
\usepackage{booktabs}         
\usepackage{amssymb}
\usepackage{amsthm}
\usepackage{algorithm}
\usepackage{algorithmic}
\newtheorem{proposition}{Proposition}

\newtheorem{theorem}{Theorem}
\newtheorem{definition}{Definition}

\newtheorem{lemma}{Lemma}
\newtheorem{remark}{Remark}

\newcommand{\E}{{\mathbb{E}}}

\begin{document}

\title{Optimal Compression of Unit Norm Vectors \\in the High Distortion Regime}
\author{%
  \IEEEauthorblockN{Heng Zhu$^*$, Avishek Ghosh$^\dagger$ and Arya Mazumdar$^*$}
  
  \IEEEauthorblockA{$^*$University of California, San Diego \\
  $^\dagger$Indian Institute of Technology, Bombay \\
  Emails: hez007@ucsd.edu, avishek$\_$ghosh@iitb.ac.in, arya@ucsd.edu
}
}

 \maketitle

\begin{abstract}
  Motivated by the need for communication-efficient distributed learning, we investigate the method for compressing a unit norm vector into the minimum number of bits, while still allowing for some acceptable level of distortion in recovery. This problem has been explored in the rate-distortion/covering code literature, but our focus is exclusively on the "high-distortion" regime. We approach this problem in a worst-case scenario, without any prior information on the vector, but allowing for the use of randomized compression maps. Our study considers both biased and unbiased compression methods and determines the optimal compression rates. It turns out that simple compression schemes are nearly optimal in this scenario. While the results are a mix of new and known, they are compiled in this paper for completeness. 
\end{abstract}

\section{Introduction}
\label{sec:intro}
Data compression is an active area of research in signal processing, specially in the context of image, audio and video processing \cite{gersho2012vector,salomon2010handbook,sayood2017introduction}. Apart from the classical applications, in modern large scale learning systems, like federated learning \cite{MAL-083,konevcny2016federated,mcmahan2017communication}, data compression plays a crucial role. Federated learning (FL) is a distributed learning paradigm, with one center and several (in millions or even more) client machines, and exchange of information between the server and the clients is crucial for learning models. 
In a generic federated learning framework, we have each of a large number of clients computing local stochastic gradients of a loss function and sending it to the server. The server averages the received stochastic gradients and updates the model parameters.
In modern machine learning, one necessarily needs to deal with extreme high dimensional data %
($d$ is very large) and train a large neural network with billions (if not more) of parameters. Hence, one of the major challenges faced by FL is communication cost between the client machines and the server since this cost is associated with the internet bandwidth of the users which are often resource constraints~\cite{MAL-083}.


A canonical way to reduce the communication cost in FL is to compress (sprasify or quantize) the data (gradients) before communicating, and this forms the basis of our study. Indeed, data compression is widely used in FL systems (\cite{alistarh2017qsgd,stich2018sparsified,karimireddy2019error,ghosh2021communication}) and several compression schemes, both deterministic as well as randomized, have been proposed in the last few years. Broadly, the compression schemes used in FL falls under two categories: (i) unbiased (\cite{alistarh2017qsgd,stich2018sparsified}) and (ii) biased (\cite{karimireddy2019error,stich2018sparsified,ghosh2021communication}), where for unbiased compressor, we require conditions on first and second moment (see Definition~\ref{def unbiased}) and for the biased one, we just require a condition on second moment (see Definition~\ref{def biased}). We now define them formally:

\begin{definition}[Unbiased $\omega$-compressor]
\label{def unbiased}
A randomized operator $\mathcal{Q}$: $\mathbb{R}^d \rightarrow \mathbb{R}^d$ is an {\emph{unbiased $\omega$-compressor}} if it satisfies
\begin{align}
\label{unbiased}
    \E_{\mathcal{Q}}[\mathcal{Q}(x)] =& x,  \notag \\
    \E_{\mathcal{Q}} \left \|\mathcal{Q}(x) - x\right \|^2 \le & (\omega-1) \left \|x\right \|^2, \quad \forall x \in \mathbb{R}^d, \notag
\end{align}
where $\omega \geq 1$ is the (unbiased) compression parameter.
\end{definition}

Typical unbiased compressors include:
\begin{itemize}
    \item Randomized Quantization in QSGD \cite{alistarh2017qsgd}: For any real number $r \in [a,b]$, with probability $\frac{b-r}{b-a}$  quantize $r$ into $a$, and with probability $\frac{r-a}{b-a}$  quantize $r$ into $b$.
    \item Rand-$k$ sparsification \cite{wangni2018gradient}: For any $x \in \mathbb{R}^d$, randomly select $k$ elements of $x$ to be scaled by $\frac{d}{k}$, and let the other elements be zero.
\end{itemize}

A more general definition including biased compressor is given by the following:
\begin{definition}[Biased $\delta$-compressor]
\label{def biased}
A randomized operator $\mathcal{Q}$: $\mathbb{R}^d \rightarrow \mathbb{R}^d$ is a $\delta$-compressor if it satisfies
\begin{align}
   \E_{\mathcal{Q}} \left \|\mathcal{Q}(x) - x\right \|^2 \le & (1-\delta)\left \|x\right \|^2, \quad \forall x \in \mathbb{R}^d, \notag
\end{align}
where $\delta \in [0,1]$ is the (biased) compression parameter. If the compressor is not random, we remove the expectation.
\end{definition}

Typical biased compressors include:
\begin{itemize}
    \item $\ell_1$-sign quantization \cite{karimireddy2019error}: For any $x \in \mathbb{R}^d$, $\mathcal{Q}(x) = \frac{\left \|x \right \|_1}{p}{\rm sign}(x)$. Here $\delta$ is $\frac{\left \|x \right \|_1^2}{p\left \|x \right \|}$.
    \item Top-$k$ sparsification \cite{stich2018sparsified}: For any $x \in \mathbb{R}^d$, select $k$ elements with the largest absolute values to be remained, and let the other elements to be zero. Here $\delta$ is $\frac{k}{d}$.
\end{itemize}
\textbf{Ranges of $\omega$ and $\delta$:} The parameters $\omega$ and $1-\delta$ measure the amount of distortion of the compressors. We emphasize that for the biased compressor, the $\delta$ is in the range of $[0,1]$, while in the unbiased case,  $\omega$ is in the range of $[1,+\infty)$. Owing to unbiasness, the variance of the compressor often increases and we require $\omega$ to be large in such settings \cite{alistarh2017qsgd,gandikota2021vqsgd}. Thus there is a fundamental difference between definition of biased $\delta$-compressor and unbiased compressor. In this work we consider the high distortion regime, where $\delta$ is small and $\omega$ is large, again motivated by FL applications where extreme compression is desirable. 

We aim to find the optimal communication cost of the above-mentioned widely used compressors. 
Since, in standard applications of FL, one can normalize the information vector and sends the norm separately, we can, without loss of generality, assume that the compressors are applied to the unit vector $x$, i.e., $\|x\|^2=1$ holds in the Definition \ref{def unbiased} and \ref{def biased}.

\subsection{Our contributions}
\label{sec:contri}
Motivated by the compression needs in FL, in this paper, we study the basic problem of compressing a unit norm vector in high dimension. Though this problem has been extensively explored in the context of rate-distortion and covering code \cite{berger2003rate,gersho2012vector,cohen1997covering}, we are interested in the high distortion regime, where the recovery error is high, and the compression rate $\to 0$. Moreover, we do not put any prior information on the vector to be compressed and hence our results hold for the worst-case setting.

We first obtain lower bounds on the number of bits required to transmit such unit norm vectors under biased as well as unbiased compression schemes. For the biased schemes, this follow directly via a sphere covering argument and for the unbiased case,  {we can leverage some results from \cite{chen2020breaking,chen2022fundamental}}. Moreover, we propose and analyze several efficient algorithms that matches this lower bound and hence optimal.  It turns out that for unbiased $\omega$-compressor the minimum number of bits required is $\Omega(d/\omega)$, whereas for biased $\delta$-compressor the least number of bits required is $\Omega(d \delta)$. For upper bounds, we  provide the following different compressors:

\paragraph{An optimal but inefficient biased $\delta$-compressor} In Section~\ref{sec:gauss_delta}, we propose and analyze an unbiased $\delta$-compressor based on generating Random Gaussian codebook. In this scheme, out of the generated codebook, we choose the Gaussian vector closest (in $\ell_2$ norm) to the information vector as the quantized vector. This scheme only requires ${O}(d\delta)$ number of bits, which matches the lower bound. Although optimal, the encoding is an exhaustive search in  exponentially large number of possible codewords, and hence this algorithm is not computationally efficient.

\paragraph{A near-optimal and efficient biased $\delta$-compressor} In Section~\ref{sec:efficient_biased}, we propose an efficient biased compressor that is (near) optimal, namely Max Block Norm Quantization (MBNQ). In this algorithm, we first find the sub-block of length $k$ that has the largest norm, and then use scalar quantization (coordinate-wise) on the sub-vector. As shown in Theorem~\ref{theorem_mbnq}, MBNQ requires ${O}(d\delta\log (d\delta))$ number of bits. Furthermore, MBNQ is computationally efficient as seen in Algorithm~\ref{MBNQ}. 


\paragraph{A near-optimal and efficient biased $\omega$-compressor} In Section~\ref{sec:inefficeint_unbiased}, we discuss the vector quantized SGD (VQSGD) algorithm of \cite{gandikota2021vqsgd}, which requires $\mathcal{O}(d/\omega)$ number of bits and hence optimal. However, similar to the Random Gaussian codebook algorithm, this is also inefficient, and the computation complexity scales exponentially with dimension $d$. In Section~\ref{sec:efficient_unbiased}, we propose an efficient algorithm, namely Sparse Randomized Quantization Scheme (SRQS), that first applies the rand-$k$ compressor (of \cite{stich2018sparsified}) and then uses QSGD (of \cite{alistarh2017qsgd}) on the sparse $k$-length sub-vector. As shown in Theorem~\ref{theorem_srqs}, this simple combination yields an efficient algorithm requiring ${O}(d/\omega)\log(d \omega)$ bits , and hence SRQS is (near) optimal.

Our contribution is summarized in Table 1.
\begin{table}[t!]
    \centering
    \caption{Summary of findings: \\
    LB: Lower Bounds, Random: Random Gaussian Codebook Scheme and VQSGD \cite{gandikota2021vqsgd}, Sparse: Max Block Norm Quantization and Sparse Randomized Quantization Scheme}
    \begin{sc}
    \begin{tabular}{cccc}
        \hline
        Compressor & LB & Random & Sparse\\
         \hline
        Biased & $\Omega(d\delta)$ & $O(d\delta)$ & $O(d\delta \log d\delta)$ \\
        Unbiased & $\Omega\left(\frac{d}{\omega}\right)$ \cite{chen2020breaking} & $O\left(\frac{d}{\omega}\right)$ \cite{gandikota2021vqsgd} & $O\left(\frac{d}{\omega} \log d\omega \right)$ \\
        \hline
    \end{tabular}
\end{sc}
\end{table}

\section{Related Work}
\label{sec:related_work}

\subsection{$\epsilon$-nets}
Note that, our biased compressors are just $\epsilon$-nets for the unit sphere ($\epsilon = 1-\delta$), and it is known that there exists an $\epsilon$-net of size $(1+2/\epsilon)^d$~\cite{vershynin2010introduction}. However, results on $\epsilon$-nets such as this are tailored for $\epsilon$ very small; for $\epsilon$ very close to 1, they do not provide a fine-grained dependence on $\delta$.
\subsection{Quantization}
Possibly, a straightforward way to cut down communication is to use quantization or sparsification techniques. Since, typically the dimension of such data is huge, dimension reduction techniques often turn out to be useful. To achieve compression, one can quantize each coordinate of a transmitted vector into few bits \cite{alistarh2017qsgd,wen2017terngrad,zhang2017zipml}, or obtain a sparser vector by letting some elements be zero \cite{stich2018sparsified,wangni2018gradient,konevcny2018randomized}. These compressors are either biased and unbiased, influencing the convergence of learning algorithms. {In \cite{mayekar2020ratq,mayekar2020limits}, a lower bound of compression budget is obtained to retain the convergence rate of stochastic optimization of the uncompressed setting. To achieve the lower bound, a compression scheme based on the random rotation and block quantization was proposed. In context of distributed mean estimation, a lower bound of estimation accuracy is derived given a compression budget \cite{suresh2017distributed}, which is based on a distributed statistical estimation lower bound. Also a random rotation scheme was proposed to approach the lower bound.}


\subsection{Federated Learning and Communication Cost}
As mentioned in the introduction, 
FL algorithms employ several techniques to reduce the communication cost. One simple way is to use local iterations and communicate to the server sparsely~\cite{mcmahan2017communication,stich2019local,karimireddy2020scaffold}. 
Another way to reduce number of communication rounds is to use second order Newton based algorithm \cite{dane,giant,ghoshsecond,ghoshcubic}, which exploits the compute power of the client machines and cut the communication cost. 

\section{Biased $\delta$-Compressors}

We start with the biased $\delta$-compressor. We first answer the question of minimum communication cost of $\delta$-compressor by providing a lower bound, which is provided by the sphere covering.
\subsection{Lower Bound}
\label{sec:lowerbound}

\begin{proposition}
\label{theorem_biased}
    Let $Q(v)$ be the compressor of $v$ satisfying Definition \ref{def biased}. Then the number of bits $b$ required to transmit $Q(v)$ satisfy
    \begin{align}
        b \geq d\delta+ \log d.
    \end{align}
\end{proposition}
\begin{proof}
    The proof employs a simple sphere covering of the unit sphere $S^{d-1}$~\cite{cohen1997covering}. 
    We know that $\mathcal{Q}{(v)}$ resides within the ball $B_d(v,1-\delta)$. The cardinality of the set of balls covering the sphere $S^{d-1}$ is at least
    \begin{align}
        C = \frac{\text{vol}(S^{d-1})}{\text{vol}(S^{d-1} \cap B_d(v,1-\delta))} \geq \frac{\text{vol}(S^{d-1})}{\text{vol}(B_d(v,1-\delta))} \notag
    \end{align}
    We can send the indices of these balls covering the unit sphere. Thus the necessary number of bits $b$ to represent $\mathcal{Q}(v)$ satisfies
    \begin{align}
        b = & \log C \geq \log \frac{\text{vol}(S^{d-1})}{\text{vol}(B_d(v,1-\delta))} = \log \frac{2\pi^{\frac{d}{2}}/\Gamma(\frac{d}{2})}{[\pi^{\frac{d}{2}}/\Gamma(\frac{d}{2}+1)] (1-\delta)^d} \notag \\
        = & \log \frac{d}{(1-\delta)^d} = \log d - d\log (1-\delta) \geq \log d + d\delta, \notag
    \end{align}
    where $\Gamma(\cdot)$ is the Gamma function, the last inequality is from the fact $\log(1+x) \leq x$ for $x>-1$.
\end{proof}
Considering that $\delta$ falls within the range of $[0,1]$, the required number of bits is smaller than transmitting the uncompressed vector directly, which would require $O(d)$ bits.

In the following we first present a random Gaussian codebook scheme that achieves the aforementioned lower bound. However, this scheme incurs significant computational and storage requirements, rendering it infeasible in practice. Consequently, we introduce an alternative practical sparse quantization scheme which is nearly optimal.

\subsection{Random Gaussian Codebook Scheme}
\label{sec:gauss_delta}
We use a vector quantization method via random Gaussian construction: let a random Gaussian matrix $A \in \mathbb{R}^{d\times n}$ be
\begin{align}
\label{gaussianmatrix}
    A = \frac{1}{\sqrt{N}} \left [ a_1, \dots, a_n \right ] 
\end{align}
where each column $a_i \in \mathbb{R}^{d}$ is a Gaussian vector with each element being standard Gaussian variable $\mathcal{N}(0,1)$. $N$ is the normalization factor, that will be chosen later. The columns of Gaussian matrix $A$ are regarded as a codebook for compression.

For a unit norm vector $v$ to be compressed, we find the nearest Gaussian vector
\begin{align}
    a_\text{min} = \underset{\{a_i\}_{i=1}^n}{\arg\min} \left \| v - \frac{1}{\sqrt{N}}a_i \right \|^2. \notag
\end{align}
Then we map $v$ to the Gaussian vector
\begin{align}
    \mathcal{Q}(v) = \frac{1}{\sqrt{N}} a_\text{min} .\notag
\end{align}

To transmit the compressed vector, we can only send the index of the nearest Gaussian vector $i_\text{min}$.  The process of random Gaussian scheme is described in Algorithm \ref{RGCS}.

\begin{algorithm}[tb]
\caption{Random Gaussian Codebook Scheme}
\label{RGCS}
\textbf{Input}: Unit norm vector $v$; A matrix $A \in \mathbb{R}^{d\times n}$ constructed as (\ref{gaussianmatrix}) \\
\textbf{Encoding}:  
\begin{algorithmic}[1]
\STATE Calculate the distance between $v$ and all the Gaussian vectors $\{ a_i\}_{i=1}^n$ : $ \text{dist}_i = \left \| v - \frac{1}{\sqrt{N}}a_i \right \|^2$
\STATE Find the smallest distance and corresponding index $i_\text{min}$
\end{algorithmic}
\textbf{Output}: Index of the Gaussian vector $i_\text{min}$
\end{algorithm}

For this scheme, we have the following guarantee:
\begin{theorem}
\label{theorem_gauss}
    A Gaussian codebook constructed as (\ref{gaussianmatrix}) with $N=\frac{d}{\delta}$ of size 
    $n = \exp( O(d\delta))$, is
    a $\delta$-compressor 
    with  probability at least $P_\delta= 1 -  \exp \left[ -\frac{n}{2\sqrt{2\pi}(2+t)\sqrt{d\delta}} \exp(-\frac{(2+t)^2 d\delta}{8}) + \frac{2d}{\sqrt{1-\delta}} \right]= 1- o(1)$. The number of bits needed is
    \(
        b = \log_2 n = O(d\delta). 
    \)
\end{theorem}

The proof is in the Appendix \ref{appen_biasedgauss}.
\begin{remark}
Since we use $O(d\delta)$ bits to represent the index of $n$ Gaussian vectors, the random Gaussian codebook scheme can achieve the lower bound of biased $\delta$-compressor.
However, it is impractical to store the large number of Gaussian vectors and perform a lot of distance computations at each iteration, where $n$ grows exponentially with model dimension $d$.   
\end{remark}

\subsection{Max Block Norm Quantization (MBNQ)}
\label{sec:efficient_biased}
To give a practical compression scheme, we show that a simple block quantization scheme can approach the lower bound, with an extra logarithmic factor. This scheme, Max Block Norm Quantization, uses standard scalar quantization on a sub-block with largest norm. 

Our scheme is first to uniformly partition the vector into sub-blocks $v_j \in \mathbb{R}^k$, $j=1,2,\dots,d/k$. Then we pick up the sub-block $v_{\max}$ with largest norm from the vectors $\{v_j\}_{j=1}^{d/k}$. 

We next quantize the $v_{\max}$ with standard scalar quantization $SQ(\cdot)$ with $l$ quantization levels to get the quantized vector $q = SQ(v_{\max})$. The quantization function is $SQ(x) = r_i$ if $x\in [low_i, up_i), i=1,\dots, l$; $low_i, up_i$ are the lower and upper bound of a quantization interval, {$r_i$ can be any quantization} value inside the quantization interval. That is, each element of $v_{\max}$ is quantized to the closest quantization value. The scalar quantization function is applied element-wise. The final compressor is shown as $\mathcal{Q}(v) = q$. The detailed process of MBNQ is described in Algorithm \ref{MBNQ}.

\begin{algorithm}[tb]
\caption{Max Block Norm Quantization (MBNQ)}
\label{MBNQ}
\textbf{Input}: Unit norm vector $v$ \\
\textbf{Initialization}: Define a scalar quantization function $SQ(\cdot)$ with $l$ quantization levels, where $SQ(x) = r_i$ if $x\in [low_i, up_i), i=1,\dots, l$, $low_i, up_i$ are the lower and upper bound of a quantization interval, $r_i$ is the quantization value inside the quantization interval
\begin{algorithmic}[1]
    \STATE Partition the vector $v$ into $d/k$ sub-vectors $v_j \in \mathbb{R}^k$, $j=1,2,\dots,d/k$
    \STATE Calculate the norms of sub-vectors $u_j = \|v_j\|^2$, then pick up the $v_\text{max}$ with largest norm $u_\text{max}$
    \STATE Use scalar quantization function to quantize $v_\text{max}$ $q=SQ(v_\text{max})$, where the quantization function is applied element-wise
\end{algorithmic}
\textbf{Output}: Compressed vector $q$
\end{algorithm}

For the MBNQ scheme, we have the following guarantee.
\begin{theorem}
\label{theorem_mbnq}
    With number of sub-blocks $k=2d\delta$ and 
    the number of quantization levels $l=\sqrt{4d\delta}$, the MBNQ is a $\delta$-compressor with the number of bits
    \begin{align}
        b =  d\delta \log (4d\delta) + \log \frac{1}{2\delta} .\notag
    \end{align}
\end{theorem}
\begin{proof}
    The compression error of our proposed scheme can be expressed as follows:
\begin{align}
\label{error_biased}
    \E \left \| v - Q(v) \right \|^2 \ = & \E \left \| v - \bar{v}_{\max} + \bar{v}_{\max} - q \right \|^2 \notag \\
    = & \E \left \| v - \bar{v}_{\max} \right \|^2 + \E \left \| \bar{v}_{\max} - q \right \|^2 ,
\end{align}
where $\bar{v}_{\max} \in \mathbb{R}^d$ represents a vector with $k$ elements equal to $v_{\max}$ and the remaining elements set to zero. The last equality is from the fact that $\bar{v}_{\max} - q$ is the vector with $k$ non-zero elements, and $v-\bar{v}_{\max}$ is the vector with the remaining $d-k$ non-zero elements, so that $\langle v-\bar{v}_{\max}, \bar{v}_{\max} - q \rangle$ would be zero.

The first term in (\ref{error_biased}) corresponds to the error resulting from the partitioning of the vector. Since $v_{\max}$ is a sub-vector of $v$, we have
\begin{align}
    \left \| v-\bar{v}_{\max} \right \|^2 = 1-u_{\max}^2. \notag
\end{align} 
where $u_{\max}$ denotes the norm of $v_{\max}$. The second term in (\ref{error_biased}) represents the error resulting from scalar quantization with $l$ levels. From the quantization error in scalar quantization, we have
\begin{align}
    \E \left \| \bar{v}_{\max} - q \right \|^2 \leq \frac{k}{l^2} \left \| \bar{v}_{\max} \right \|^2 = \frac{k}{l^2} u_{\max}^2. \notag
\end{align}
Here we choose $k = \frac{l^2}{2}$, then we can obtain
\begin{align}
    \E \left \| v - Q(v) \right \|^2 \leq & 1-u_{\max}^2 + \frac{k}{l^2}u_{\max}^2 =  1 - \frac{u_{\max}^2}{2}. \notag
\end{align}
Note the error decreases with a larger $u_{\max}$. Since we pick up the sub-vector with largest norm, we have $u_{\max}^2 \geq \frac{k}{d}$. Therefore we can obtain
\begin{align}
    \E \left \| v - Q(v) \right \|^2 \leq 1 - \frac{k}{2d}. \notag
\end{align}
Hence, the MBNQ scheme achieves the $\delta$-compressor with $\delta=\frac{k}{2d}$.

Regarding the communication cost of the MBNQ scheme, each worker needs to transmit the index of the sub-vector with the largest norm and $k$ quantized real numbers. Thus, the total number of required bits can be calculated as:
\begin{align*}
    b = & \log \frac{d}{k} + k\log l =  \log \frac{d}{k} + \frac{k}{2}\log (2k) =  \log \frac{1}{2\delta} + d\delta \log (4d\delta), \notag
\end{align*}
which represents the communication cost in Theorem \ref{theorem_mbnq}.
\end{proof}

\begin{remark}
    The implementation of MBNQ is straightforward and easy in real systems. And the simple MBNQ scheme can approach the lower bound of biased $\delta$-compressor, except for a logarithmic factor.
    If we directly apply scalar quantization on the whole vector $v$, the communication cost would be $O(d)$. The improvement is from the quantization on a sub-vector. When the size of sub-vector $k$ is chosen properly, the quantization can get nearly optimal performance.
\end{remark}

\section{Unbiased $\omega$-Compressors}

In this section we provide results on the unbiased compressor of definition \ref{def unbiased}. We start with a lower bound and a random codebook scheme to achieve this lower bound, both  following  from existing results. Next we propose a practical and simple sparse quantization scheme to approach the lower bound, except for a logarithmic factor.

\subsection{Lower Bound}
For unbiased compressor, \cite{chen2020breaking} already provided a lower bound under both communication and privacy constraints. The lower bound is proven by constructing a prior distribution on $v$ and analyzing the compression error. Here we directly give the lower bound in \cite{chen2020breaking}. The lower bound also follows from~\cite{gandikota2021vqsgd}.

\begin{theorem}[\cite{chen2020breaking}, Appendix D.2]
\label{theorem_unbiased}
    Let $Q(v)$ be the compressor of $v$ satisfying definition \ref{def unbiased}. Then the number of bits $b$ required to describe $Q(v)$ satisfies
    \begin{align}
        b = \Omega\left(\frac{d}{\omega}\right). 
    \end{align}
\end{theorem}

  Note that, $\omega$ is in the range of $[1,+\infty)$. Thus the lower bound of $\Omega(\frac{d}{\omega})$ is smaller than directly transmitting the vector as $\Omega(d)$ real numbers.

\subsection{Unbiased Random Gaussian Codebook Scheme}
\label{sec:inefficeint_unbiased}
In \cite{gandikota2021vqsgd}, a randomized vector quantization framework is proposed for unbiased compression. Among the quantization methods introduced in \cite{gandikota2021vqsgd}, there is a random Gaussian codebook method to achieve the lower bound of unbiased compressor.

Similar to the  Random Gaussian Codebook Scheme in the biased case, we first construct a random Gaussian matrix $A \in \mathbb{R}^{d\times n}$:
\begin{align}
\label{unbiasedgaussianmatrix}
    A = \frac{1}{\sqrt{N}} \left [ a_1, \dots, a_n \right ] 
\end{align}
where each column $a_i \in \mathbb{R}^{d\times 1}$ is a Gaussian vector with elements as standard Gaussians $\mathcal{N}(0,1)$. $N$ is the normalization factor. The columns of Gaussian matrix $A$ are regarded as a codebook for compression, i.e., the points in $C$ are $c_i = \frac{1}{\sqrt{N}} a_i, \quad i = 1,\dots,n$.
Then we use the Gaussian vectors to perform the compression:
\begin{align}
    Q(v) = \frac{1}{\sqrt{N}} a_i, \text{with probability} \quad p_i. \notag
\end{align}
The detailed process of unbiased random Gaussian codebook scheme is described in Algorithm \ref{URGCS}.

\begin{algorithm}[tb]
\caption{Unbiased Random Gaussian Codebook Scheme}
\label{URGCS}
\textbf{Input}: Unit norm vector $v$; A matrix $A \in \mathbb{R}^{d\times n}$ constructed  as (\ref{unbiasedgaussianmatrix})\\
\textbf{Encoding}: 
\begin{algorithmic}[1]
\STATE Construct a linear convex combination of Gaussian vectors $\{\frac{1}{\sqrt{N}} a_i \}_{i=1}^n$ to get $v$: $v = \sum\limits_{i=1}^n \frac{1}{\sqrt{N}} p_i a_i$
\STATE Randomly choose from $\{a_i\}_{i=1}^n$ with probability distribution $\{p_i\}_{i=1}^n$: $Q(v) = \frac{1}{\sqrt{N}} a_i, \text{with probability} ~ p_i$ 
\STATE Get the index of the chosen Gaussian vector $i$
\end{algorithmic}
\textbf{Output}: Index of the chosen Gaussian vector $i$
\end{algorithm}
\begin{theorem}[\cite{gandikota2021vqsgd}, Theorem 7]
\label{theorem_unbiasedgauss}
    A Gaussian codebook of size $n = \exp( O(\frac{d}{\omega} + \log d))$, constructed as in (\ref{unbiasedgaussianmatrix}), with $N=\frac{9d}{\omega}$, $\omega \in [25,36d]$, is an unbiased compressor (definition \ref{def unbiased}), with probability $1-o(1)$. The number of compressed bits  is
    \begin{align}
        b = \log n = O\left(\frac{d}{\omega} + \log d\right). \notag
    \end{align}
\end{theorem}
\begin{remark}
    From the above theorem, the dominiating term is $O(\frac{d}{\omega})$ bits, matching the lower bound in Theorem \ref{theorem_unbiased}. However, the unbiased Gaussian codebook scheme also has the problem of impractical computational and storage burden.
\end{remark}
\subsection{Sparse Randomized Quantization Scheme (SRQS)}
\label{sec:efficient_unbiased}

To design a practical compression method to approach the lower bound, we propose a Sparse Randomized Quantization Scheme (SRQS) based on the well-known rand-$k$ sparsification and randomized quantization in QSGD~\cite{alistarh2017qsgd}. It is easy to implement and we prove it can approach the lower bound except for a logarithmic factor.

Our SRQS method is as follows. The rand-$k$ compressor randomly chooses $k$ elements of the vector $v$ and let other elements be zero. Specifically, for each coordinate $j$ with value $v_j$, with probability $\frac{k}{d}$ the value is set to be $\frac{d}{k} v_j$, and probability $1-\frac{k}{d}$ to be 0. 

For the sparse vector $v_\text{spa}$ obtained from $v$, we perform an unbiased Randomized Quantization as  in \cite{alistarh2017qsgd} to quantize the vector $q=RQ(v_\text{spa})$. This method, called QSGD, works as follows. The quantization function consists of $l$ quantization levels. For each element $v_j$ in the $v_\text{spa}$, the Randomized Quantization function is 
\begin{align}
    RQ(v_j) = \| v_\text{spa} \|_2 \cdot \text{sign}(v_j) \cdot \xi_j(v_\text{spa}, l) \notag
\end{align}
where $\text{sign}(x) \in \{+1,-1\}$ represents the sign of $x$, $\xi_j$ is an independent random variable that determines which quantization level that $v_j$ is mapped to, as defined next. Let $0 \leq t < l$ be an integer such that $\frac{|v_j|}{\|v_\text{spa}\|_2} \in [\frac{t}{l}, \frac{t+1}{l}]$, i.e., $[\frac{t}{l}, \frac{t+1}{l}]$ is the quantization interval for $v_j$. Then the random variable $\xi_j$ is 
\begin{align}
    \xi_j(v_\text{spa}, l) =  \left \{
    \begin{array}{cc}
    (t+1)/l, & \text{with probability} \quad \frac{|v_j|}{\|v_\text{spa}\|} \cdot l - t\\
    t/l,     & \text{otherwise},
    \end{array} \right. \notag
\end{align}
Thus the Randomized Quantization is an unbiased method $\E [RQ(v_\text{spa})] = v_\text{spa}$. 
As in QSGD, the final quantized vector $q$ is expressed by a tuple $(\left\|v_\text{spa}\right\|, s, z)$, where vector $s$ includes the signs of the non-zero elements, $s_i \in \{+1,-1\}$, and $z$ includes the quantization levels of non-zero elements, i.e., $z_j = \xi_j \cdot l \in \{0, 1, \dots, l\}$. 

Overall, the the two-stage compression scheme sequentially applying rand-$k$ and Randomized Quantization, is unbiased.  It can be seen as a more sparse version of QSGD. The detailed process of SRQS is described in Algorithm \ref{SRQS}.

\begin{algorithm}[tb]
\caption{Sparse Randomized Quantization Scheme}
\label{SRQS}
\textbf{Input}: Unit norm vector $v$ \\
\textbf{Initialization}: Design a Randomized Quantization function $RQ(\cdot)$ with $l$ quantization levels
\begin{algorithmic}[1]
    \STATE Apply rand-$k$ compressor on $v$ to output $v_\text{spa}$: For each coordinate $v_j$, there is a probability $\frac{k}{d}$ to be $\frac{d}{k} v_j$, and probability $1-\frac{k}{d}$ to be 0
    \STATE Use Randomized Quantization function on $v_\text{spa}$ to obtain the compressed vector $q = RQ(v_\text{spa})$, the $q$ is expressed a tuple $(\left\|v_\text{spa}\right\|, s, z)$.
\end{algorithmic}
\textbf{Output}: Compressed vector $q$
\end{algorithm}

Now we give a lemma on the accuracy of SRQS.
\begin{lemma}
\label{lemma_accuracy}
    The SRQS achieves compression error
    \begin{align}
        \E \left \| v - \mathcal{Q}(v) \right \|^2 \leq \left[ \frac{d}{k} (1+\frac{k}{l^2}) -1 \right] \left \| v\right \|^2, \notag
    \end{align}
    with rand-$k$ compressor and Randomized Quantization with $l$ quantization levels.
\end{lemma}
\begin{proof}
Please note that our scheme incorporates two sources of randomness: the rand-k sparsification and the unbiased Randomized Quantization. Let $\E_{\mathcal{Q}}$ denote the expectation over Randomized Quantization. The total compression error of the SRQS scheme can be expressed as follows:
\begin{align}
\label{error_sqrs}
    & \E \left \| v-\mathcal{Q}(v) \right \|^2 =  \E \left \| v-v_{spa}+v_{spa}-q \right \|^2 \notag \\
    = & \E ( \left \| v-v_{spa} \right \|^2 + 2 \left<  v-v_{spa}, v_{spa}-q \right> + \left \| v_{spa}-q \right \|^2 )\notag \\
    = & \E (\left \| v-v_{spa} \right \|^2 + 2 \E_Q \left< v-v_{spa}, v_{spa}-q \right> + \E_Q \left \| v_{spa}-q \right \|^2 ) \notag \\
    = & \E \left(\left \| v-v_{spa} \right \|^2 + \E_Q \left \| v_{spa}-q \right \|^2 \right) 
\end{align}
where the last equality arises from the unbiased property $\E_Q (v_{spa}-q) = 0$. The error in the second term of (\ref{error_sqrs}) is from the Randomized Quantization. According to Lemma 3.1 in QSGD \cite{alistarh2017qsgd}, we obtain the quantization error as:
\begin{align}
    \E_Q \left \| v_{spa} - q \right \|^2 \leq \min \{ \frac{k}{l^2}, \frac{\sqrt{k}}{l} \} \left \|v_{spa}\right \|^2 \notag
\end{align}
Here we choose $l=\sqrt{2k}$, which ensures $\frac{k}{l^2}\leq 1$. Consequently, we can deduce that $\E_Q \left \| v_{spa} - q \right \|^2 \leq  \frac{k}{l^2} \left \|v_{spa}\right \|^2$.

Then, the total compression error can be expressed as follows:
\begin{align}   
    & \E \left \| v-\mathcal{Q}(v) \right \|^2 \leq  \E( \left \| v-v_{spa} \right \|^2 + \frac{k}{l^2}\left \|v_{spa} \right \|^2 ) \notag \\
    = & \E \left(\left \| v \right \|^2 - 2 \left<v,v_{spa}\right> + (1+\frac{k}{l^2})\left\|v_{spa}\right\|^2 \right) \notag \\
    = & (1+\frac{k}{l^2}) \E \left \|v_{spa}\right\|^2 - \left \|v \right \|^2 \notag
\end{align}
where the last equality is from the unbiased property of rand-$k$ sparsification, $\E v_{spa} = v$.

Given that $v_{spa}$ is obtained through the rand-k sparsification, we have:
\begin{align}
    \E \left \|v_{spa}\right \|^2 = \sum\limits_{j=1}^d \frac{k}{d} (\frac{d}{k} v_j)^2 = \sum\limits_{j=1}^d \frac{d}{k} v_j^2 = \frac{d}{k} \left \|v\right \|^2. \notag
\end{align}
Finally we can get the total compression error as:
\begin{align}
    \E \left \| v - \mathcal{Q}(v) \right \|^2 \leq \left[ \frac{d}{k} (1+\frac{k}{l^2}) -1 \right] \left \| v\right \|^2. \notag
\end{align}
\end{proof}
This lemma indicates that in our scheme, the compression error is given by 
$\omega=\frac{d}{k}(1+\frac{k}{l^2})$. 

To transmit the final compressed vector $q$, we encode and transmit the tuple $(\left\|v_\text{spa}\right\|, s, z)$. In our scheme, we employ Elias coding, similar to QSGD, for this purpose. Elias coding is a method used to encode positive integers in the process, such as the locations of non-zero elements in $q$ and the integer vector $z$.

Now let us present the formal guarantee of SRQS.

\begin{theorem}
\label{theorem_srqs}
    When $\omega \geq 9$, $k=\frac{3d}{2\omega}$, quantization level $l=\sqrt{2k}$, the SRQS is unbiased $\omega$-compressor with the number of bits
   $ 
        b = O(\frac{d}{\omega}\log d \omega). \notag
   $
\end{theorem}

The proof is in the Appendix \ref{appen_srqs}.

    As can be seen, the SRQS approaches the lower bound of unbiased compressor, except for a logarithmic factor $\log d\omega$.


\section{Conclusion}
In this paper we compile the number of necessary and sufficient bits required for widely used biased $\delta$-compressor and unbiased $\omega$-compressors. For biased $\delta$-compressor, we propose a random Gaussian codebook scheme to achieve the lower bound, and a Max Block Norm Quantization scheme to approach the lower bound up to a logarithmic factor. For unbiased compressor, we also show an unbiased random Gaussian codebook scheme can achieve the lower bound. And we further propose a practical Sparse Randomized Quantization Scheme to approach the lower bound, up to a logarithmic factor.In short, an application of the simple combination of sparsification and quantization methods on distributed learning leads to near-optimal compression.

\vspace{0.1in}
\noindent\emph{Acknowledgement:} This work is supported in part by NSF awards 2133484, 2217058, and 2112665.

\bibliography{federated}

\begin{thebibliography}{10}

\bibitem{gersho2012vector}
Allen Gersho and Robert~M Gray,
\newblock {\em Vector quantization and signal compression}, vol. 159,
\newblock Springer Science \& Business Media, 2012.

\bibitem{salomon2010handbook}
David Salomon, Giovanni Motta, and Giovanni Motta,
\newblock {\em Handbook of data compression}, vol.~2,
\newblock Springer, 2010.

\bibitem{sayood2017introduction}
Khalid Sayood,
\newblock {\em Introduction to data compression},
\newblock Morgan Kaufmann, 2017.

\bibitem{MAL-083}
Peter Kairouz and H~Brendan McMahan,
\newblock ``Advances and open problems in federated learning,''
\newblock {\em Foundations and Trends® in Machine Learning}, vol. 14, no. 1,
  pp. 1--210, 2021.

\bibitem{konevcny2016federated}
Jakub Kone{\v{c}}n{\`y}, H~Brendan McMahan, Felix~X Yu, Peter Richt{\'a}rik,
  Ananda~Theertha Suresh, and Dave Bacon,
\newblock ``Federated learning: Strategies for improving communication
  efficiency,''
\newblock {\em arXiv preprint arXiv:1610.05492}, 2016.

\bibitem{mcmahan2017communication}
Brendan McMahan, Eider Moore, Daniel Ramage, Seth Hampson, and Blaise~Aguera
  y~Arcas,
\newblock ``Communication-efficient learning of deep networks from
  decentralized data,''
\newblock in {\em Artificial intelligence and statistics}. PMLR, 2017, pp.
  1273--1282.

\bibitem{alistarh2017qsgd}
Dan Alistarh, Demjan Grubic, Jerry Li, Ryota Tomioka, and Milan Vojnovic,
\newblock ``{QSGD}: Communication-efficient {SGD} via gradient quantization and
  encoding,''
\newblock in {\em Advances in Neural Information Processing Systems}, 2017, pp.
  1709--1720.

\bibitem{stich2018sparsified}
Sebastian~U Stich, Jean-Baptiste Cordonnier, and Martin Jaggi,
\newblock ``Sparsified {SGD} with memory,''
\newblock in {\em Advances in Neural Information Processing Systems}, 2018, pp.
  4447--4458.

\bibitem{karimireddy2019error}
Sai~Praneeth Karimireddy, Quentin Rebjock, Sebastian~U Stich, and Martin Jaggi,
\newblock ``Error feedback fixes {S}ign{SGD} and other gradient compression
  schemes,''
\newblock in {\em International Conference on Machine Learning}, 2019, pp.
  3252--3261.

\bibitem{ghosh2021communication}
Avishek Ghosh, Raj~Kumar Maity, Swanand Kadhe, Arya Mazumdar, and Kannan
  Ramchandran,
\newblock ``Communication-efficient and byzantine-robust distributed learning
  with error feedback,''
\newblock {\em IEEE Journal on Selected Areas in Information Theory}, vol. 2,
  no. 3, pp. 942--953, 2021.

\bibitem{wangni2018gradient}
Jianqiao Wangni, Jialei Wang, Ji~Liu, and Tong Zhang,
\newblock ``Gradient sparsification for communication-efficient distributed
  optimization,''
\newblock in {\em Advances in Neural Information Processing Systems}, 2018, pp.
  1299--1309.

\bibitem{gandikota2021vqsgd}
Venkata Gandikota, Daniel Kane, Raj~Kumar Maity, and Arya Mazumdar,
\newblock ``vqsgd: Vector quantized stochastic gradient descent,''
\newblock in {\em International Conference on Artificial Intelligence and
  Statistics}. PMLR, 2021, pp. 2197--2205.

\bibitem{berger2003rate}
Toby Berger,
\newblock ``Rate-distortion theory,''
\newblock {\em Wiley Encyclopedia of Telecommunications}, 2003.

\bibitem{cohen1997covering}
G{\'e}rard Cohen, Iiro Honkala, Simon Litsyn, and Antoine Lobstein,
\newblock {\em Covering codes},
\newblock Elsevier, 1997.

\bibitem{chen2020breaking}
Wei-Ning Chen, Peter Kairouz, and Ayfer Ozgur,
\newblock ``Breaking the communication-privacy-accuracy trilemma,''
\newblock {\em Advances in Neural Information Processing Systems}, vol. 33, pp.
  3312--3324, 2020.

\bibitem{chen2022fundamental}
Wei-Ning Chen, Christopher A~Choquette Choo, Peter Kairouz, and Ananda~Theertha
  Suresh,
\newblock ``The fundamental price of secure aggregation in differentially
  private federated learning,''
\newblock in {\em International Conference on Machine Learning}. PMLR, 2022,
  pp. 3056--3089.

\bibitem{vershynin2010introduction}
Roman Vershynin,
\newblock ``Introduction to the non-asymptotic analysis of random matrices,''
\newblock {\em arXiv preprint arXiv:1011.3027}, 2010.

\bibitem{wen2017terngrad}
Wei Wen, Cong Xu, Feng Yan, Chunpeng Wu, Yandan Wang, Yiran Chen, and Hai Li,
\newblock ``Terngrad: Ternary gradients to reduce communication in distributed
  deep learning,''
\newblock in {\em Advances in Neural Information Processing Systems}, 2017, pp.
  1509--1519.

\bibitem{zhang2017zipml}
Hantian Zhang, Jerry Li, Kaan Kara, Dan Alistarh, Ji~Liu, and Ce~Zhang,
\newblock ``Zipml: Training linear models with end-to-end low precision, and a
  little bit of deep learning,''
\newblock in {\em International Conference on Machine Learning}, 2017, pp.
  4035--4043.

\bibitem{konevcny2018randomized}
Jakub Kone{\v{c}}n{\`y} and Peter Richt{\'a}rik,
\newblock ``Randomized distributed mean estimation: Accuracy vs.
  communication,''
\newblock {\em Frontiers in Applied Mathematics and Statistics}, vol. 4, no.
  62, 2018.

\bibitem{mayekar2020ratq}
Prathamesh Mayekar and Himanshu Tyagi,
\newblock ``Ratq: A universal fixed-length quantizer for stochastic
  optimization,''
\newblock in {\em International Conference on Artificial Intelligence and
  Statistics}. PMLR, 2020, pp. 1399--1409.

\bibitem{mayekar2020limits}
Prathamesh Mayekar and Himanshu Tyagi,
\newblock ``Limits on gradient compression for stochastic optimization,''
\newblock in {\em 2020 IEEE International Symposium on Information Theory
  (ISIT)}. IEEE, 2020, pp. 2658--2663.

\bibitem{suresh2017distributed}
Ananda~Theertha Suresh, X~Yu Felix, Sanjiv Kumar, and H~Brendan McMahan,
\newblock ``Distributed mean estimation with limited communication,''
\newblock in {\em International conference on machine learning}. PMLR, 2017,
  pp. 3329--3337.

\bibitem{stich2019local}
Sebastian~U Stich,
\newblock ``Local {SGD} converges fast and communicates little,''
\newblock in {\em International Conference on Learning Representations}, 2019.

\bibitem{karimireddy2020scaffold}
Sai~Praneeth Karimireddy, Satyen Kale, Mehryar Mohri, Sashank Reddi, Sebastian
  Stich, and Ananda~Theertha Suresh,
\newblock ``Scaffold: Stochastic controlled averaging for federated learning,''
\newblock in {\em International Conference on Machine Learning}. PMLR, 2020,
  pp. 5132--5143.

\bibitem{dane}
Ohad Shamir, Nathan Srebro, and Tong Zhang,
\newblock ``Communication efficient distributed optimization using an
  approximate newton-type method,''
\newblock {\em CoRR}, vol. abs/1312.7853, 2013.

\bibitem{giant}
Shusen Wang, Farbod Roosta-Khorasani, Peng Xu, and Michael~W. Mahoney,
\newblock ``Giant: Globally improved approximate newton method for distributed
  optimization,'' 2017.

\bibitem{ghoshsecond}
Avishek Ghosh, Raj~Kumar Maity, and Arya Mazumdar,
\newblock ``Distributed newton can communicate less and resist byzantine
  workers,''
\newblock in {\em Proceedings of the 34th International Conference on Neural
  Information Processing Systems}, Red Hook, NY, USA, 2020, NIPS'20, Curran
  Associates Inc.

\bibitem{ghoshcubic}
Avishek Ghosh, Raj~Kumar Maity, Arya Mazumdar, and Kannan Ramchandran,
\newblock ``Escaping saddle points in distributed newton's method with
  communication efficiency and byzantine resilience,''
\newblock {\em CoRR}, vol. abs/2103.09424, 2021.

\end{thebibliography}
\bibliographystyle{IEEEbib}


\appendices

\section{Proof of Theorem \ref{theorem_gauss}}
\label{appen_biasedgauss}
\begin{proof}
    We begin by considering a specific fixed unit vector $v$. Let $\gamma^2 = 1-\delta$. For the $i$-th column $a_i$ in $A$, the probability of the event that the distance between $v$ and $a_i$ exceeds $\gamma$ can be expressed as follows:
\begin{align}
    & \text{Pr} \left ( \left \| v- \frac{1}{\sqrt{N}} a_i \right\|^2 > \gamma^2 \right ) = \text{Pr} \left ( 1 + \frac{1}{N} \|a_i\|^2 -\frac{2}{\sqrt{N}} \left<v,a_i\right> > \gamma^2 \right ). \notag
\end{align}
Since $\left \| v\right \|^2=1$, $a_i$ is the vector with each element being a standard Gaussian variable, the inner-product $\left<v,a_i\right>$ is a Gaussian variable. We denote it as $z = \left<v,a_i\right>$ and $ z \sim \mathcal{N}(0,1)$. Thus, we can write the probability as:
\begin{align}
    & \text{Pr} \left ( 1+ \frac{1}{N} \|a_i\|^2 -\frac{2}{\sqrt{N}} z > \gamma^2 \right ) = \text{Pr} \left( z < \frac{\sqrt{N}}{2} \left(1+ \frac{1}{N} \|a_i\|^2-\gamma^2 \right) \right) \notag 
\end{align}
Since $\|a_i\|^2$ follows chi-squared distribution with degree $d$, we can get the tail bound of $\|a_i\|^2$ as
\begin{align}
    \text{Pr} \left( \left| \frac{1}{d}\|a_i\|^2 - 1 \right| \geq t \right) \leq 2 \exp (-\frac{dt^2}{8}) \notag
\end{align}
We define the event $A$ as $A = \left\{ z < \frac{\sqrt{N}}{2} \left(1+ \frac{1}{N} \|a_i\|^2-\gamma^2 \right) \right\}$, event $B_1$ as $B_1=\left\{ \left| \frac{1}{d}\|a_i\|^2 - 1 \right| \geq t \right\}$, and $B_2=\left\{ \left| \frac{1}{d}\|a_i\|^2 - 1 \right| < t \right\}$. Then we have
\begin{align}
    \text{Pr} (A) = \text{Pr} (A, B_1) + \text{Pr} (A, B_2) \leq \text{Pr} (B_1) + \text{Pr}(A, B_2) \notag
\end{align}
From event $\{A,B_2\}$, we can imply that $z < \frac{\sqrt{N}}{2} \left(1+ \frac{d}{N} (1+t)-\gamma^2 \right)$. Thus we know that
\begin{align}
    &\text{Pr} (B_1) \leq 2 \exp (-\frac{dt^2}{8}) \notag \\
    &\text{Pr}(A, B_2) \leq \text{Pr} \left( z < \frac{\sqrt{N}}{2} \left(1+ \frac{d}{N} (1+t)-\gamma^2 \right) \right) = 1- \text{Pr} \left( z \geq \frac{\sqrt{N}}{2} \left(1+ \frac{d}{N} (1+t)-\gamma^2 \right) \right). \notag
\end{align}
Let $E = \frac{\sqrt{N}}{2}(1+ \frac{d}{N}(1+t)-\gamma^2)$. By choosing $N=\frac{d}{\delta}$ and recalling that $\gamma^2 = 1-\delta$, we obtain $E = (1+\frac{t}{2})\sqrt{d\delta}$.

For the standard Gaussian variable $z$, the tail bound is
\begin{align}
\text{Pr} (z\geq x) \geq \frac{C}{x} e^{-x^2/2},  \quad \text{when} \quad x \geq 1  .\notag
\end{align}
where $C = \frac{1}{2\sqrt{2\pi}}$. We can see $E>1$ from the chosen N, thus we can have
\begin{align}
    &\text{Pr} \left( z \geq E  \right) \geq  \frac{C}{E} \exp (-\frac{E^2}{2}). \notag
\end{align}
Combining above results, we have
\begin{align}
    & \text{Pr} \left ( \left \| v- \frac{1}{\sqrt{N}} a_i \right\|^2 > \gamma^2 \right ) \leq 1 - \frac{C}{E} \exp(-\frac{(2+t)^2d\delta}{8} ) + 2\exp (-\frac{dt^2}{8}). \notag
\end{align}
In our setting, $\delta$ is small and can approach 0, thus we can choose a constant $t$ to satisfy $(2+t)^2 \delta < t^2$. We can further let 
\begin{align}
     2\exp (-\frac{dt^2}{8}) \leq \frac{C}{2E} \exp(-\frac{(2+t)^2d\delta}{8} ) \notag
\end{align}
Taking logrithm on both sides, we can get
\begin{align}
\label{condition1}
    4 \log d\delta + 8 \log \frac{2+t}{2} + 4\log 128\pi \leq d [t^2 - (2+t)^2\delta]
\end{align}
Since the left-hand side of (\ref{condition1}) is at order $O(\log d\delta)$, and the right-hand side is at order $d$. When we choose t to satisfy $(2+t)^2\delta < t^2$ and $d$ is large, the condition (\ref{condition1}) can easily hold.

Then we can obtain
\begin{align}
    \text{Pr} \left ( \left \| v- \frac{1}{\sqrt{N}} a_i \right\|^2 > \gamma^2 \right ) \leq 1 - \frac{C}{2E} \exp (-\frac{(2+t)^2d\delta}{8} ) . \notag
\end{align}
Since there are $n$ i.i.d. random Gaussian vectors, we define the event that for all Gaussian vectors, there is no close vector to the particular $v$:
\begin{align}
    \mathcal{E}_v = \left \{ \forall i \in [n], \quad \left \| v- \frac{1}{\sqrt{N}} a_i \right\|^2 > \gamma^2 \right \}. \notag
\end{align}
The probability of event $\mathcal{E}_v$ is
\begin{align}
    \text{Pr}[\mathcal{E}_v] & \leq  \left[ 1 - \frac{C}{2E}\exp (-\frac{(2+t)^2d\delta}{8} )  \right]^n \notag \\
    & \leq  \exp \left[ -\frac{C}{2E}n\exp (-\frac{(2+t)^2d\delta}{8} ) \right], \notag
\end{align}
where the second inequality is from the fact that $1-x\leq e^{-x}$.

To encompass all possible $v$ on the unit sphere, we employ an $\epsilon$-net to cover the unit sphere, where we set the error parameter $\epsilon$ equal to $\gamma$. Consequently, the size of the $\epsilon$-net is given by $\left ( 1+\frac{2}{\gamma}\right)^d$ \cite{vershynin2010introduction}. Subsequently, we define the event that for all Gaussian vectors, there are no close vectors to any of the vectors in the $\epsilon$-net:
\begin{align}
    \mathcal{E} = \left \{ \forall i \in [n], \forall v \in \epsilon ~\text{net}, \quad \left \| v- \frac{1}{\sqrt{N}} a_i \right\|^2 > \gamma^2 \right \}. \notag
\end{align}
By union bound, the possibility of event $\mathcal{E}$ is
\begin{align}
    \text{Pr}[\mathcal{E}] \leq & \left ( 1+\frac{2}{\gamma}\right)^d \exp \left[ -\frac{C}{2E}n\exp (-\frac{(2+t)^2d\delta}{8} )  \right]  \notag \\
    \leq & \exp(\frac{2d}{\gamma}) \exp \left[ -\frac{C}{2E}n\exp (-\frac{(2+t)^2d\delta}{8} )  \right]  \notag \\
    = & \exp \left[ \frac{2d}{\gamma} -\frac{C}{2E}n\exp (-\frac{(2+t)^2d\delta}{8} )\right], \notag
\end{align}
where the second inequality is from the fact that $1+x\leq e^x$.

Here we expect that the exponential term is negative, so that $\text{Pr}[\mathcal{E}]$ approaches 0. Hence, we need to ensure the following condition holds:
\begin{align}
    \frac{2d}{\gamma} \leq \frac{C}{2E}n\exp (-\frac{(2+t)^2d\delta}{8} ). \notag
\end{align}
Taking logarithm on both sides, we get
\begin{align}
\label{ineq1}
    \log n \geq & \frac{(2+t)^2d\delta}{8} +\log 4d -\log \gamma - \log \frac{C}{E} \notag \\
    = &  \frac{(2+t)^2d\delta}{8} + \frac{1}{2}\log \left( \frac{32\pi d^3\delta(2+t)^2}{1-\delta} \right). \notag
\end{align}
When 
$
    \log n \geq O(d\delta) \notag
$
 holds, the probability $\text{Pr}[\mathcal{E}]$ approaches to 0. Then the communication cost of the random Gaussian codebook scheme is 
\begin{align}
    b = \log n = O(d\delta). \notag
\end{align}
The scheme is a $\delta$-compressor with  probability at least
\begin{align}
    1 - \text{Pr}[\mathcal{E}] = 1 -  \exp \left[ -\frac{n}{2\sqrt{2\pi}(2+t)\sqrt{d\delta}} \exp(-\frac{(2+t)^2 d\delta}{8}) + \frac{2d}{\sqrt{1-\delta}} \right] \notag
\end{align}
\end{proof}

\section{Proof of Theorem \ref{theorem_srqs}}
\label{appen_srqs}

During the transmission of compressed vectors, we use the Elias coding to encode positive integers. Before proving Theorem \ref{theorem_srqs}, we need to introduce a lemma from \cite{alistarh2017qsgd} that demonstrates the number of bits needed to represent a vector after Elias coding.
\begin{lemma} [\cite{alistarh2017qsgd}, Appendix, Lemma A.3]
\label{lemma_elias}
Let $y\in \mathbb{N}^d$ be a vector of which each element $y_i$ is a positive integer, and its $\ell_p$-norm is $\left \| y \right \|^p_p \leq \rho$, then we have
\begin{align}
    \sum\limits_{i=1}^d | \text{Elias}(y_i) | \leq \left( \frac{1+o(1)}{p} \log \frac{\rho}{d} + 1 \right) d
\end{align}
where $\text{Elias}(\cdot)$ is the Elias coding function applied to a positive integer. 
\end{lemma}

Now we give the proof of Theorem \ref{theorem_srqs}.
\begin{proof}
    After rand-$k$ sparsification and Randomized Quantization, the compressed vector $q$ becomes very sparse. To transmit $q$, we first need to send the locations of $\|q\|_0$ non-zero elements. Let $i_1, i_2, \dots, i_{\|q\|_0}$ represent the non-zero indices of $q$. We use Elias coding to encode the integer vector $[i_1, i_2-i_1, \dots, i_{\|q\|_0} - i_{\|q\|_0 -1}]$. This integer vector has a length of $\|q\|_0$ and an $\ell_1$-norm at most $d$. Thus from Lemma \ref{lemma_elias} the required number of bits after Elias coding is given by
\begin{align}
    b_1 = \left( (1+o(1)) \log \frac{d}{\|q\|_0} + 1 \right) \|q\|_0 .\notag
\end{align}
Secondly, we need to transmit the sign vector $s$ and the integer vector $z$. For sign vector $s$ with $\|q\|_0$ elements, we need $b_2 = \|q\|_0$ bits. For the integer vector $z$, we also use Elias coding to encode it. The required number of bits after Elias coding is 
\begin{align}
    b_3 = \left( \frac{1+o(1)}{2} \log \frac{\left\|z\right\|_2^2}{\|q\|_0} + 1 \right) \|q\|_0. \notag
\end{align}

From \cite{alistarh2017qsgd}, we know for Randomized Quantization, the number of non-zero elements is
\begin{align}
    \E \|q\|_0 \leq l^2 + \sqrt{\|v_{spa}\|_0} \notag
\end{align}
and the squared norm of $z$ is 
\begin{align}
    \|z\|^2_2 \leq 2(l^2+\|v_{spa}\|_0) .\notag
\end{align}
Summing up the three parts $b_1, b_2, b_3$, we can get the total number of bits
\begin{align}
    \E b = 3 \E \|q\|_0 + \E (1+o(1)) \|q\|_0 \left( \log \frac{d}{\|q\|_0} + \frac{1}{2}\log\frac{2(l^2+\|v_{spa}\|_0)}{\|q\|_0} \right) \notag
\end{align}
Note that the function $x\log\frac{C}{x}$ increases until $x=\frac{C}{2}$ and then decreases. Also function $x\log\frac{C}{x}$ is concave so that $\E \left[ x\log\frac{C}{x} \right] \leq \E x \log \frac{C}{\E x}$. Assuming $l^2+\|v_{spa}\|_0 \leq \frac{d}{2}$, it follows that $l^2+\sqrt{\|v_{spa}\|_0} \leq \frac{d}{2}$. Applying $x = l^2+\sqrt{\|v_{spa}\|_0}$, $C=d$ in the function, we can obtain
\begin{align}
    \E b \leq & 3 \E \|q\|_0 + \E (1+o(1)) \|q\|_0 \left( \frac{3}{2} \log \frac{d}{\|q\|_0}  \right) \notag \\
    \leq & 3 \E\left(l^2+\sqrt{\|v_{spa}\|_0}\right) + \E(1+o(1))\left(l^2+\sqrt{\|v_{spa}\|_0}\right) \left( \frac{3}{2} \log \frac{d}{l^2+\sqrt{\|v_{spa}\|_0}} \right) \notag
\end{align}
where the first inequality is from $l^2+\|v_{spa}\|_0 \leq \frac{d}{2}$, the second inequality is from the property of function $x\log\frac{C}{x}$.

From the rand-$k$ sparsification, we know that $\E (l^2 + \sqrt{\|v_{spa}\|_0}) \leq l^2 + \sqrt{k}$. Assume $l^2 + k \leq \frac{d}{2}$. Applying the property of function $x\log\frac{C}{x}$ again, and let $x=l^2 + \sqrt{\|v_{spa}\|_0}$, $C=d$, we can have
\begin{align}
    \E b \leq & 3 (l^2+\sqrt{k}) + \frac{3}{2}(1+o(1))(l^2+\sqrt{k}) \log \frac{d}{l^2+\sqrt{k}} \notag \\
    = & (l^2+\sqrt{k}) \left[ 3 + \frac{3}{2} (1+o(1)) \log \frac{d}{l^2+\sqrt{k}} \right]. \notag
\end{align}

From Lemma \ref{lemma_accuracy}, our scheme has parameter $\omega = \frac{d}{k} (1+\frac{k}{l^2})$. Here We choose $l = \sqrt{2k}$, then we have $\omega = \frac{3}{2}\frac{d}{k}$.

Finally the communication cost of SQRS is
\begin{align}
    \E b \leq (2k + \sqrt{k}) \left[ 3 + \frac{3}{2} (1+o(1)) \log \frac{d}{2k+\sqrt{k}} \right]. \notag
\end{align}
The dominating term is $O(k\log \frac{d}{\sqrt{k}})$, i.e., $O(\frac{d}{\omega}\log d\omega)$.

To satisfy the condition $l^2 + \|v_{spa}\|_0 \leq l^2 + k \leq \frac{d}{2}$, we require
\begin{align}
    l^2 + k = 3k \leq \frac{d}{2}. \notag
\end{align}
Thus $k \leq \frac{d}{6}$, and it follows that $\omega \geq 9$.

When $\omega \geq 9$, and we choose $k=\frac{3d}{2k}$, $l=\sqrt{2k}$, then we can achieve unbiased compressor with communication cost $O(\frac{d}{\omega}\log d\omega)$ bits.
\end{proof}

\end{document}